\documentclass{article}
\usepackage{graphicx} 
\usepackage{amsmath,amsfonts,amssymb,amsthm}
\usepackage{physics}
\usepackage[
  left=1.25in,
  right=1.25in,
  top=1in,
  bottom=1in,
]{geometry}
\usepackage{authblk}
\usepackage[english]{babel}
\usepackage[colorlinks,allcolors=blue]{hyperref}
\usepackage[sorting=none]{biblatex}
\title{No cardinality bound for squashed entanglement }

\author[1,2]{Rabsan Galib Ahmed}
\author[1,2]{Graeme Smith}
\affil[1]{Department of Applied Mathematics, University of Waterloo, Ontario, Canada}
\affil[2]{Institute for Quantum Computing, University of Waterloo, Ontario, Canada}
\date{September 12, 2026}
\newtheorem{theorem}{Theorem}

\newtheorem{lemma}[theorem]{Lemma}

\begin{document}

\maketitle
\begin{abstract}
    Squashed entanglement is an additive bipartite entanglement measure. For a state $\rho_{AB}$, it is defined as the infimum of all conditional mutual information $I(A:B\mid E)$ evaluated on extensions $\rho_{ABE}$. It was unresolved whether there is a bound on the dimension of the conditioning system, $E$, needed for this optimization. We show that no such cardinality bound is possible. Explicitly, we find that the squashed entanglement for a partially dephased Bell pair on two qubits cannot be achieved for any finite dimensional extension $E$. A crucial ingredient of this proof is a direct sum step, which in a sense, symmetrizes two purifying systems while keeping the conditional mutual information unchanged and the coherence no worse.   
\end{abstract}

\section{Preliminaries}
For a tripartite state on $ABE$, the \textit{conditional mutual information} (CMI) between systems $A$ and $B$, conditioned on $E$ is defined as follows
\begin{align}
    I(A:B\mid E) = S(AE) + S(BE) - S(ABE) - S(E).
\end{align}
Among many entanglement monotones available in the literature, squashed entanglement stands out because it is additive, faithful, and continuous~\cite{Christandl2004-sn,Brand_o_2011,Alicki2004-xl}. The squashed entanglement of a bipartite state $\rho_{AB}$ is defined as
\begin{align}
    E_{\text{sq}}(\rho_{AB}) = \inf \left\{\frac{1}{2}I(A:B\mid E)_{\omega_{ABE}}\text{ such that $\Tr_E \omega_{ABE} = \rho_{AB}$}\right\}.
\end{align}
It is not known whether for an arbitrary bipartite state on finite dimension, this infimum can be achieved with an extension on a finite dimensional conditioning system $E$~\cite{Beigi_2014}. In this article we illustrate a counterexample to any possibility of a cardinality bound.  

First we define a \textit{dimensionally bounded squashed entanglement}, as: 
$E_{D}(\rho_{AB}) = \inf \{\frac{1}{2}I(A:B\mid E)_{\omega_{ABE}}\text{ such that $\Tr_E \omega_{ABE} = \rho_{AB}$ and $\dim E \leq D$}\}.$
By definition, $E_{\text{sq}}(\rho_{AB}) = \inf_{D\geq 1} E_D(\rho_{AB})$. In the following we give an example of a state, $\rho$, such that for all $D \ge 1$ there exists $D'>D$ such that $E_{D'}(\rho) < E_D(\rho)$. This implies that the infimum is never achieved for any finite dimensional extension.

The following trace inequality will be useful for proving the main result, so we include it here. Closely related inequalities can be found in~\cite{Araki1981-hi,Kittaneh1986-wt}. 
\begin{lemma}\label{lem: tr-ineq}
    For any two matrices $A_0$ and $A_1$ of the same order,
    \begin{align}
        2 |\Tr (A_0 A_1^\dagger)| \leq \Tr (\sqrt{A_0A_0^\dagger} \sqrt{A_1 A_1^\dagger} + \sqrt{A_0^\dagger A_0} \sqrt{A_1^\dagger A_1}).
    \end{align}
\end{lemma}
    \begin{proof}
        First, using the AM-GM inequality, we obtain
        \begin{align}
             &\Tr (\sqrt{A_0A_0^\dagger} \sqrt{A_1 A_1^\dagger} + \sqrt{A_0^\dagger A_0} \sqrt{A_1^\dagger A_1})\nonumber\\ \geq &\;2 \left[\Tr (\sqrt{A_0A_0^\dagger} \sqrt{A_1 A_1^\dagger})\Tr (\sqrt{A_0^\dagger A_0} \sqrt{A_1^\dagger A_1})\right]^{1/2}.
        \end{align}
        Now consider the singular value decompositions, $A_i = U_i D_i V_i$. Then the two trace terms above are reduced to 
        \begin{align}
            \Tr (\sqrt{A_0A_0^\dagger} \sqrt{A_1 A_1^\dagger}) = \Tr (U_0 D_0U_0^\dagger U_1D_1U_1^\dagger) = \Tr (\underbrace{\sqrt{D_0}U_0^\dagger U_1\sqrt{D_1}}_{X^\dagger} \underbrace{\sqrt{D_1}U_1^\dagger U_0\sqrt{D_0}}_{X})\nonumber\\
            \Tr (\sqrt{A_0^\dagger A_0} \sqrt{A_1^\dagger A_1}) = \Tr (V_0^\dagger D_0V_0 V_1^\dagger D_1V_1) = \Tr (\underbrace{\sqrt{D_0}V_0 V_1^\dagger\sqrt{D_1}}_{Y^\dagger} \underbrace{\sqrt{D_1}V_1 V_0^\dagger\sqrt{D_0}}_{Y}).\nonumber
        \end{align}
        Applying the Cauchy-Schwarz inequality, $|\Tr (Y^\dagger X)|\leq [\Tr (X^\dagger X) \Tr (Y^\dagger Y)]^{1/2}$, we obtain
        \begin{align}
            &\left[\Tr (\sqrt{A_0A_0^\dagger} \sqrt{A_1 A_1^\dagger})\Tr (\sqrt{A_0^\dagger A_0} \sqrt{A_1^\dagger A_1})\right]^{1/2} \nonumber\\
            \geq & |\Tr (\sqrt{D_0}V_0 V_1^\dagger\sqrt{D_1} \sqrt{D_1}U_1^\dagger U_0\sqrt{D_0})|\nonumber\\
            = & |\Tr (U_0 D_0 V_0 V_1^\dagger D_1U_1^\dagger)| = |\Tr (A_0 A_1^\dagger)|.
        \end{align}
        Combining everything, we have the claimed trace inequality.
    \end{proof}

First developed in~\cite{Siddhu_2021}, entropic singularity has emerged as a useful tool in quantum information theory~\cite{Singh_2022,Platypus}. To summarise, consider states $\rho$ and $\sigma$ acting on a Hilbert space $\mathcal{H}$ with $\mathrm{rank} \;\rho  < \dim \mathcal{H}$ while $\mathrm{rank}\;\sigma = \dim \mathcal{H}$, i.e., $\sigma$ is full rank while $\rho$ is rank-deficient. Moreover, for some $0<t<1$, consider the state $\rho_t = (1-t)\rho + t\sigma$. Entropic singularity is the following observation: $S(\rho_t) = S(\rho) + (\Tr{K\sigma}) \; t\log (1/t) + O(t)$, where $K$ is the projector on $\ker \rho$. For sufficiently small $t$, the $t\log (1/t)$ term dominates over $O(t)$ which leads to higher gain in entropy for rank deficient states in comparison to the full rank states upon such mixing. This constitutes a crucial step in the argument. 

All logarithms are in base $2$ and $h_2(p):= -p\log p - (1-p)\log (1-p)$ denotes the binary Shannon entropy.
\section{Main result}
For $0<c<1$, consider the following state on the two qubits $A$ and $B$:
\begin{align}
    \rho_{c} = \frac{1}{2}\left(\ketbra{00}{00} + c\ketbra{00}{11} + c\ketbra{11}{00} + \ketbra{11}{11}\right).
\end{align}
This is a dephased Bell pair on two qubits with a coherence parameter, $c$, and constitutes our counterexample, as proved below.
\begin{theorem}
 For all $0<c<1$ and positive integers $D$, $E_{3D+2} (\rho_c) < E_{D}(\rho_c)$.   
\end{theorem}
\begin{proof}
 Let $\omega_{ABE}$ be an extension of $\rho_c$ by a $D$-dimensional system, $E$. As $\rho_c$ has support on a $2$ dimensional subspace of the two-qubit space, we can purify $\omega_{ABE}$ by a system $F$ such that $\dim F = 2D$. Clearly, we can write such a purification as
    \begin{align}
        \ket{\omega}_{ABEF} = \ket{00}_{AB} |M_0\rangle\!\rangle_{EF} + \ket{11}_{AB}|M_1\rangle\!\rangle_{EF}.
    \end{align}
    Here $|M_i\rangle\!\rangle_{EF} = (M_i\otimes \mathbb{I}_F) \sum_{j=1}^{2D} \ket{jj}_{F'F}$, for some linear map $M_i: F'\to E$, with $F'\cong F$.

    Using the identities: $\Tr_F{| A \rangle \!\rangle \langle\!\langle B|} = AB^\dagger$ and $\Tr_E{| A \rangle \!\rangle \langle\!\langle B|} = (B^\dagger A)^T$, we can write down the relevant reduced states needed for the evaluation of the CMI:
    \begin{align}
        \omega_{AE} &= \ketbra{0}{0}_A \otimes M_0M^\dagger_0 + \ketbra{1}{1}_A \otimes M_1M_1^\dagger,\\
        \omega_{BE} &= \ketbra{0}{0}_B \otimes M_0M^\dagger_0 + \ketbra{1}{1}_B \otimes M_1M_1^\dagger,\\
        \omega_{E} & = M_0M^\dagger_0 + M_1M^\dagger_1,\\
        \omega_F & = (M^\dagger_0 M_0 + M^\dagger_1M_1)^T.
    \end{align}
    Furthermore, $\omega_{AB} = \rho_c$ implies that $\Tr M^\dagger_0 M_0 = \Tr M^\dagger_1 M_1 = \frac{1}{2}$ and $\Tr M^\dagger_1 M_0 = \Tr M^\dagger_0 M_1=\frac{c}{2}$. Therefore, we can evaluate the CMI as
    \begin{align}
        &I(A:B\mid E)_{\omega}\nonumber\\ = \;&S(AE) + S(BE) - S(ABE) - S(E)\nonumber\\
        = \;&S(AE) + S(BE) - S(F) - S(E)\nonumber\\
        = \;&2S(M_0M^\dagger_0 \oplus M_1M^\dagger_1) - S(M^\dagger_0 M_0 + M^\dagger_1M_1) - S(M_0M^\dagger_0 + M_1M^\dagger_1).
    \end{align}
    For the third equality we use the fact that as $\ket{\omega}_{ABEF}$ is a pure state, $S(ABE)=S(F)$. We note that $(I:B\;|\;E)_{\omega} > 0 $. This is because
    \begin{align}
        I(A:B\mid E)_{\omega} &= S(AE) + S(BE) - S(F) - S(E)\nonumber\\
        &= S(AE) + S(AF) - S(F) - S(E)\nonumber\\
        &=S(A|E) + S(A|F) \geq 2 S(A|EF).
    \end{align}
    The last inequality follows from strong subadditivity. However $S(A|EF) = S(AEF) - S(EF) = S(B) - S(AB) = 1-h_2((1+c)/2) >0$, as $c>0$. 

    In the next stage of the proof, we construct another extension of $\rho_c$ on $\Tilde{E}$ with $\dim \Tilde{E} = 3D + 2$. First, consider the state
    \begin{align}
        \ket{\phi}_{ABGH} = \ket{00}_{AB} | P_0\rangle\!\rangle_{GH} +\ket{11}_{AB} | P_1\rangle\!\rangle_{GH}, 
    \end{align}
    where $A,B$ are two qubits and $\dim G = \dim H = 3D$, and 
    \begin{align}
        P_i = \frac{1}{\sqrt{2}} \left(\sqrt{M_iM_i^\dagger} \oplus \sqrt{M_i^\dagger M_i}\right).
    \end{align}
    Now $P_i=P_i^\dagger \geq 0$, $\Tr (P_i^\dagger P_i) = \Tr (P_i^2) = \frac{1}{2}$, and $0\leq \Tr (P_0^\dagger P_1) \leq \sqrt{\Tr(P_0^2)\Tr(P_1^2)} = \frac{1}{2}$. Therefore, the reduced state $\phi_{AB} = \rho_{c'}$ with $c' = 2\Tr (P^\dagger_0P_1)$, i.e., $0\leq c' \leq 1$. Using Lemma~\ref{lem: tr-ineq} for the matrices $M_0$ and $M_1$, we have that $c' = 2\Tr(P^\dagger_0P_1) \geq 2 |\Tr (M_0M^\dagger_1)| = c$. 
    
    Furthermore, we claim that $I(A:B\; |\; E)_{\omega} = I(A:B\mid G)_{\phi}$. Indeed, the relevant reduced states are
    \begin{align}
        \phi_{AG} &= \ketbra{0}{0}_A \otimes \frac{1}{2}\left(M_0M_0^\dagger \oplus M_0^\dagger M_0\right) + \ketbra{1}{1}_A\otimes \frac{1}{2}\left(M_1M_1^\dagger \oplus M_1^\dagger M_1\right),\nonumber\\
        \phi_{BG} &= \ketbra{0}{0}_B \otimes \frac{1}{2}\left(M_0M_0^\dagger \oplus M_0^\dagger M_0\right) + \ketbra{1}{1}_B\otimes \frac{1}{2}\left(M_1M_1^\dagger \oplus M_1^\dagger M_1\right),\nonumber\\
        \phi_G &= \frac{1}{2}\left[(M_0M_0^\dagger+M_1M_1^\dagger )\oplus (M_0^\dagger M_0 + M_1^\dagger M_1)\right],\nonumber\\
        \phi_H &= \frac{1}{2}\left[(M_0M_0^\dagger+M_1M_1^\dagger )^T\oplus (M_0^\dagger M_0 + M_1^\dagger M_1)^T\right].
    \end{align}
    Now, \begin{align}
        S(AG) = S(BG) &= S\left(\frac{1}{2}(M_0^\dagger M_0\oplus M_1^\dagger M_1)\oplus \frac{1}{2}(M_0M_0^\dagger\oplus M_1M_1^\dagger)\right)\nonumber\\
        & = 1 + \frac{1}{2}S(M_0^\dagger M_0\oplus M_1^\dagger M_1) + \frac{1}{2}S(M_0 M_0^\dagger\oplus M_1 M_1^\dagger)\nonumber\\
        & = 1 + S(M_0^\dagger M_0\oplus M_1^\dagger M_1),
    \end{align}
    \begin{align}
        S(G) = S(H) &= S\left(\frac{1}{2}(M_0M_0^\dagger+M_1M_1^\dagger )\oplus \frac{1}{2}(M_0^\dagger M_0 + M_1^\dagger M_1)\right)\nonumber\\
        &= 1 + \frac{1}{2}S(M_0M_0^\dagger+M_1M_1^\dagger) + \frac{1}{2}S(M_0^\dagger M_0 + M_1^\dagger M_1).
    \end{align}
    Therefore,
    \begin{align}
        I(A:B\mid G)_\phi &= S(AG) + S(BG) - S(H) - S(G)\nonumber\\
        &= 2 S(M_0^\dagger M_0\oplus M_1^\dagger M_1) - S(M^\dagger_0 M_0 + M^\dagger_1M_1) - S(M_0M^\dagger_0 + M_1M^\dagger_1)\nonumber\\
        &= I(A:B\mid E)_\omega.
    \end{align}
It is worthwhile to note that in terms of $P_0$ and $P_1$, the CMI is given by 
\begin{align}
    I(A:B\mid G)_\phi = 2S(P_0^2\oplus P_1^2) - 2S(P_0^2+P_1^2).
\end{align}

To construct an extension of $\rho_c$ which attains a lower CMI than $I(A:B\mid E)_{\omega}$, we distinguish three cases: (i) $c' > c$, (ii) $c'=c$ with $\ker P_0 \neq \ker P_1$, and (iii) $c'=c$ with $\ker P_0 = \ker P_1$.

\vspace{4mm}
\noindent
\textbf{Case 1 ($c'>c$):} Any dephased Bell state with a lower coherence can be obtained by an appropriate convex mixture of a dephased Bell state with a higher coherence and a completely dephased Bell state. In our case
\begin{align}
    \rho_c = \frac{c}{c'}\rho_{c'} + \left(1- \frac{c}{c'}\right) \rho_{0}.
\end{align}
Take the extension $\eta_{ABG'} = \frac{1}{2}(\ketbra{00}{00}_{AB} \otimes \ketbra{0}{0}_{G'}+ \ketbra{11}{11}_{AB}\otimes \ketbra{1}{1}_{G'})$ of $\rho_0$ on some $G'\cong \mathbb{C}^2$. We then obtain an extension of $\rho_{c}$ on $\Tilde{E} =  G\oplus G'\cong \mathbb{C}^{3D+2}$
\begin{align}
    \tilde{\omega}_{AB\Tilde{E}} = \frac{c}{c'}\phi_{ABG} \oplus \left(1-\frac{c}{c'}\right)\eta_{ABG'}.
\end{align}
The CMI on this extension evaluates to
\begin{align}
    I(A:B\mid \tilde{E})_{\tilde{\omega}} &= \frac{c}{c'}I(A:B\mid G)_{\phi} + \left(1-\frac{c}{c'}\right)\underbrace{I(A:B\mid G')_{\eta}}_{=0}\nonumber\\
    &= \frac{c}{c'}I(A:B\mid E)_\omega < I(A:B\mid E)_\omega.
\end{align}
Thus, in this case, we have constructed an extension of $\rho_c$ on $(3D+2)$-dimension that has a strictly lower CMI than the extension we started with on $D$-dimension.

\vspace{4mm}
\noindent
\textbf{Case 2 ($c'=c$ with $\ker P_0 \neq \ker P_1$):} Without loss of generality, we assume that there exists a unit vector $\ket{v}\in \ker P_0$ such that $a:= \bra{v}P_1\ket{v} >0$. We will add a small support to $P_0$ along $\ket{v}$ which makes $c'$ slightly larger than $c$. Explicitly, we define $P_\epsilon = \frac{P_0+\epsilon \ketbra{v}{v}}{\sqrt{1+2\epsilon^2}}$. The normalisation is chosen to ensure that $\Tr(P_{\epsilon}^2) =\frac{1}{2}$. We have,
\begin{align}
    c'_\epsilon = 2\Tr(P_\epsilon P_1) = \frac{c + 2a\epsilon}{\sqrt{1+2\epsilon^2}} = c +2a\epsilon + o(\epsilon).
\end{align}
For sufficiently small $\epsilon$, $c'_\epsilon>c$. The CMI is given by
\begin{align}
    I_\epsilon = 2S(P_\epsilon^2\oplus P_1^2) - 2S(P_\epsilon^2+P_1^2).
\end{align}
As, $P^2_\epsilon-P_0^2 = \frac{\epsilon^2(\ketbra{v}{v}-2P_0^2)}{1+2\epsilon^2}$ and $\ket{v}\in \ker P_0$, the following two trace distance calculation follows immediately: $\frac{1}{2}\norm{(P_\epsilon^2\oplus P_1^2) - (P_0^2\oplus P_1^2)}_1 = \frac{\epsilon^2}{1+2\epsilon^2} = \frac{1}{2}\norm{(P_\epsilon^2 + P_1^2) - (P_0^2 + P_1^2)}_1$. Then using the Audenaert-Fannes' continuity bound for the von Neumann entropy~\cite{Audenaert_2007}, we obtain
\begin{align}
    \abs{I_\epsilon - I(A:B\mid G)_{\phi}} \leq 2\frac{\epsilon^2}{1+2\epsilon^2} [\log (6D-1)+\log (3D-1)] + 4h_2\left(\frac{\epsilon^2}{1+2\epsilon^2}\right).
\end{align}
For sufficiently small positive $\epsilon$, the right hand side is bounded by a constant multiple of $\epsilon^2 \log(1/\epsilon)$ due to the dominant binary entropy term. Furthermore, as $\epsilon\log(1/\epsilon)\to 0$ as $\epsilon\to 0$, the right hand side vanishes faster than $\epsilon$ as $\epsilon\to 0$, i.e., we can write that $I_\epsilon \leq I(A:B\mid G)_{\phi} + o(\epsilon)$.

Continuing with the construction as in the Case 1, we can write the convex mixture, $\rho_c = \frac{c}{c_{\epsilon}'}\rho_{c_{\epsilon}'} + \left(1- \frac{c}{c_{\epsilon}'}\right) \rho_{0}$. Using an extension of $\rho_c$ on $\Tilde{E}$, a $(3D+2)$-dimensional space, the new CMI evaluates to
\begin{align}
    I(A:B\mid \tilde{E})_{\tilde{\omega}} = \frac{c}{c_{\epsilon}'}I_\epsilon &\leq \left(1+\frac{2a\epsilon}{c} + o(\epsilon)\right)^{-1}(I(A:B\mid G)_{\phi} + o(\epsilon))\nonumber\\
    &= I(A:B\mid G)_{\phi} - \frac{2a\epsilon}{c}I(A:B\mid G)_{\phi} + o(\epsilon).
\end{align}
Therefore, with sufficiently small positive $\epsilon$, we can ensure that $I(A:B\mid \tilde{E})_{\tilde{\omega}} < I(A:B\mid G)_{\phi} = I(A:B\mid E)_\omega$. This finishes the desired construction for this case.

\vspace{4mm}
\noindent
\textbf{Case 3 ($c'=c$ with $\ker P_0 = \ker P_1$):} In this case, we redefine $P_i$'s as their restriction to their support. Therefore, the new $P_i$'s are positive definite operators acting on an $N\leq 3D$-dimensional space, $G$. Recall that the state $\rho_c$ has support in a two dimensional subspace, $\mathcal{M}\subset {A\otimes B}$. Hence the extension $\phi_{ABG}$ lives in a $2N$ dimensional subspace $\mathcal{M}\otimes G$. However, in the basis $\{\ket{00},\ket{11}\}$, we can write $\phi_{ABG}$ in the following block matrix form:
\begin{align}
    \phi_{ABG} = \begin{pmatrix}
        P_0^2 & P_0 P_1\\
        P_1P_0 & P_1^2
    \end{pmatrix} = \begin{pmatrix}
        P_0\\
        P_1
    \end{pmatrix} \begin{pmatrix}
        P_0 & P_1
    \end{pmatrix}.
\end{align}
Therefore on one hand, $\mathrm{rank}\; \phi_{ABG} = \mathrm{rank}\; P_i = N < \dim (\mathcal{M} G)=2N$. On the other hand $\mathrm{rank}\; \phi_{AG} = \mathrm{rank} \; \phi_{BG} = 2N = \dim (AG)=\dim (BG)$ and $\mathrm{rank}\; \phi_G = N = \dim G$.

Now for some $0<t<1$, consider $\omega^t_{ABG} = (1-t)\phi_{ABG} + t\rho_c\otimes \frac{\mathbb{I}_G}{N}$. Note that $\omega^t_{ABG}$ is an extension of $\rho_c$ and $\rho_c\otimes \frac{\mathbb{I}_G}{N}$ is full rank on $\mathcal{M}G$. Denote by $K$ the projector onto the kernel of $\phi_{ABG}$ in $\mathcal{M}G$, and let $b:= \Tr[K(\rho_c\otimes \frac{\mathbb{I}_G}{N})]>0$. Then the log-singularity phenomenon of entropy dictates that $S(ABG)_{\omega^t} = S(ABG)_\phi + b\; t\log (1/t) + O(t)$, whereas $S(AG)_{\omega^t} = S(AG)_\phi + O(t)$, $S(BG)_{\omega^t}= S(BG)_{\phi} + O(t)$, and $S(G)_{\omega^t} = S(G)_\phi + O(t)$. Therefore, the CMI is evaluated as
\begin{align}
    I(A:B\mid G)_{\omega^t} &= S(AG)_{\omega^t} + S(BG)_{\omega^t} - S(ABG)_{\omega^t}- S(G)_{\omega^t}\nonumber\\
    &= S(AG)_\phi + S(BG)_{\phi} - S(ABG)_\phi - S(G)_\phi - b\; t\log (1/t) + O(t)\nonumber\\
    &= I(A:B\mid G)_{\phi} - b\; t\log (1/t) + O(t).
\end{align}
Hence, for sufficiently small $t$, we can ensure that $I(A:B\mid G)_{\omega^t} < I(A:B\mid G)_{\phi} = I(A:B\mid E)_\omega$. Thus, we have found an extension of $\rho_c$ on an
$N$-dimensional conditioning system, where $N\leq 3D+2$, whose CMI is strictly lower than that of the
extension with which we started.

For a fixed finite dimension $D$, compactness of the set of states and continuity of CMI implies that $E_D(\rho_c)$ is attained. Applying the preceding construction to the minimizing $D$-dimensional extension produces an extension with conditioning
dimension at most $3D+2$ while attaining a strictly lower CMI. Therefore, $E_{3D+2}(\rho_c)<E_D(\rho_c)$ for every
$D\geq1$.
\end{proof}

\section{Conclusion}
We have shown that there is no cardinality bound for squashed entanglement. In fact for a partially dephased two-qubit maximally entangled state, we see that the squashed entanglement cannot be achieved in any finite dimensional extension. Explicitly we argue that starting from any extension of this state on $D$-dimension, we can construct an extension on $3D+2$-dimensions that achieves a strictly lower conditional mutual information. 

\section{Acknowledgements}
The counterexample and the argument was found in interaction with LLM. The paper is written by hand, and reflects our human understanding of the solution. This work is
supported under NSERC-NSF alliance grant ALLRP-586858-2023 and NSERC Discovery grant
RGPIN-2025-02094. RGA acknowledges the support of Institue for Quantum Computing and Mike and Ophelia Lazaridis Graduate Fellowship.

\printbibliography
\end{document}